\documentclass[letterpaper, 10pt, conference]{ieeetran}

\usepackage[utf8]{inputenc} 
\IEEEoverridecommandlockouts                              
\usepackage{amssymb,amsmath,amsfonts,mathrsfs}
\usepackage{amsthm}
\usepackage[pdftex,colorlinks]{hyperref}

\DeclareMathAlphabet{\mathpzc}{OT1}{pzc}{m}{it}
\newtheorem{theorem}{Theorem}[section]
\newtheorem{remark}{Remark}[section]

\newcommand{\cH}{\mathcal{H}}

\newcommand{\cS}{\mathcal{S}}

\newcommand{\cM}{\mathcal{M}}

\newcommand{\cZ}{\mathcal{Z}}

\newcommand{\bM}{\mathbf{M}}

\newcommand{\bbR}{\mathbb{R}}

\newcommand{\bbE}{\mathbb{E}}

\newcommand{\bbN}{\mathbb{N}}

\newcommand{\bbP}{\mathbb{P}}

\newcommand{\ket}[1]{\left|#1\right\rangle}

\newcommand{\tr}[1]{\mathrm{Tr}\left\{#1 \right\}}
\newcommand{\lsim}{\mathrel{\hbox{\rlap{\lower.55ex \hbox{$\sim$}} \kern-.3em \raise.4ex \hbox{$<$}}}}

\newcommand{\hrho}{\hat{\rho}}
\newcommand{\mur}{\mu^{\tt rl}}
\newcommand{\mui}{\mu^{\tt id}}
\newcommand{\mr}{m^{\tt rl}}
\newcommand{\mi}{m^{\tt id}}
\newcommand{\rhoe}{\rho^{\tt e}}
\renewcommand{\i}{\iota}

\newcommand{\ba}{\text{\bf{a}}}
\newcommand{\bn}{\text{\bf{N}}}
\newcommand{\bid}{\text{\bf{I}}}

\title{Design and Stability of Discrete-Time Quantum Filters with  Measurement Imperfections
\thanks{This work has been partially supported by the ANR under the projects CQUID, QUSCO-INCA and EPOQ2, and by EU under the IP project AQUTE and the ERC project DECLIC} }
\author{Abhinav Somaraju
\thanks{Abhinav Somaraju, Brussels Photonics Team, Dept. of Applied Physics and Photonics, Vrije Universiteit Brussel, Pleinlaan 2-1050 Brussels, Belgium. {\tt\small a.somaraju@gmail.com}}
\and Igor Dotsenko
\thanks{Igor Dotsenko, Laboratoire Kastler Brossel, \'{E}cole Normale Sup\'{e}rieure, CNRS, Universit\'{e} P. et M. Curie, 24 rue Lhomond, 75231 Paris Cedex 05, France {\tt\small igor.dotsenko@lkb.ens.fr}}
\and Clement Sayrin
\thanks{Clement Sayrin, Laboratoire Kastler Brossel, \'{E}cole Normale Sup\'{e}rieure, CNRS, Universit\'{e} P. et M. Curie, 24 rue Lhomond, 75231 Paris Cedex 05, France {\tt\small clement.sayrin@lkb.ens.fr}}
\and  Pierre Rouchon
\thanks{Pierre Rouchon, Mines ParisTech, Centre Automatique et Syst\`{e}mes,
Unit\'{e} Math\'{e}matiques et Syst\`{e}mes,
60 Bd Saint-Michel, 75272 Paris cedex 06, France
{\tt\small pierre.rouchon@mines-paristech.fr}}%
 }

\begin{document}
\maketitle
\begin{abstract} This work considers the theory underlying a  discrete-time quantum filter recently used in a quantum feedback experiment.  It proves that  this  filter taking into account decoherence and measurement errors  is optimal and stable.   We present the general framework underlying this filter and  show that it corresponds to  a recursive  expression of  the least-square optimal estimation of the density operator in the presence of  measurement imperfections.  By measurement imperfections, we mean in a very general sense unread measurement performed by the environment (decoherence) and active measurement performed by non-ideal detectors. However, we assume to know precisely all the Kraus operators and also the detection error rates.  Such recursive expressions combine  well known methods from quantum filtering theory and classical probability theory (Bayes' law). We then demonstrate  that such a  recursive filter is always stable with respect to its initial condition: the fidelity between the optimal filter state (when the initial filter state coincides with the real quantum state)  and the filter state (when the initial filter state is arbitrary)  is a sub-martingale.
\end{abstract}

\section{Introduction}

The theory of filtering considers the estimation of  the system state from noisy and/or partial observations (see, e.g.,~\cite{Bensoussan1992}). For quantum systems, filtering theory  was initiated in the 1980s by Belavkin in a series of papers~\cite{Belavkin1980,Belavkin1983,Belavkin1988,Belavkin1992} (also see the tutorial papers~\cite{Bouten2009,Bouten2007} for a more recent introduction). Belavkin makes use of  the operational formalisms of Davies~\cite{Davies1976}, which is a precursor to the theory of quantum filtering. He has also realized that due to the unavoidable back-action  of  quantum measurements, the theory of filtering plays a fundamental role in quantum feedback control (see e.g.~\cite{Belavkin1992,Belavkin1983}). The theory of quantum filtering was independently developed in the physics community, particularly in the context of quantum optics, under the name of Quantum Measurement Theory~\cite{Braginsky1992,Haroche2006,Gardiner2000,Wiseman2009}.

Most of this theory has been  developed for continuous-time systems and little emphasis has been given to measurement imprecisions and their explicit impact on the filter design and time-recursive equations. To our knowledge, the problem of designing a quantum filter in the presence of classical measurement imperfections has not been examined in the discrete time setting. In this paper, we focus on this issue  and propose a systematic method to derive quantum filters taking into account several  detection error rates.

 In~\cite[Sec. 2.2.2]{Gardiner2000}, the authors discuss how the state of a quantum system evolves after a single imprecise measurement. In~\cite{Dotsenko2009}, a recursive quantum state estimation with  measurement imperfections has been considered. In~\cite{Sayrin2011}, such a quantum state estimate  has been  used in a quantum feedback  experiment that stabilizes photon-number states of a quantized field mode,  trapped in a super-conducting cavity.
We prove here that such estimates are in fact optimal since they coincide with the conditional expectation of the quantum state (density matrix) knowing the past detections, the error rates  and the initial quantum state.

Section~\ref{sec:measModel} describes the structure of a genuine  quantum  measurement model including detection error rates which is a  straightforward   generalization of the models considered in~\cite{Gardiner2000,Dotsenko2009,Sayrin2011}. This model may be  used in situations with partial knowledge of all the quantum jumps and also measurement errors of the jumps that are detected. However, they  assume to know precisely all the Kraus operators and also the detection error rates.

Section~\ref{sec:recEqn} is devoted to the first result in this paper summarized in Theorem~\ref{the:mainRes}: the conditional expectation of the quantum state  knowing the past detections and the initial state obeys a  recursive equation in each discrete time-step. This recursive equation is  given in~\eqref{eqn:mainRes}  and  depends  explicitly on the error rates.  The proof of Theorem~\ref{the:mainRes} shows that such  recursive equation may be derived by a simple application of Bayes' law.

In section~\ref{sec:stab}, we prove that the quantum filter defined in Theorem~\ref{the:mainRes} is stable versus its initial conditions: the fidelity between the optimal estimate conditioned on the initial state of the system being known and a second estimate in which the initial state is unknown is a sub-martingale. This stability result combines Theorem~\ref{the:mainRes} and~\cite{Rouchon2011}.  Note that stability does not imply convergence, in general. For convergence results in the continuous-time case see, e.g., \cite{Handel2009} and the references therein.

In section~\ref{sec:example}, we  describe in detail the Kraus operators and error rates  modeling the  discrete-time quantum system considered in the quantum feedback experiment~\cite{Sayrin2011}: the quantum filter used in the feedback loop corresponds  precisely to the recursive equation~\eqref{eqn:mainRes} given by Theorem~\ref{the:mainRes}; according to Theorem~\ref{the:stab}, this filter tends to forget its initial condition.
\subsection{Acknowledgments} 
The authors thank  Michel Brune, Serge Haroche, Mazyar Mirrahimi and  Jean-Michel Raimond for useful discussions and advices.

\section{Measurement Model}\label{sec:measModel}
In this section we discuss the  model  describing repeated and  imperfect  measurements on a quantum system. Such  modeling including  decoherence-induced quantum jumps and measurement errors is  a direct generalization of the one proposed   in~\cite{Dotsenko2009} and used in real-time for the quantum feedback experiments reported in~\cite{Sayrin2011} (also see~\cite{Gardiner2000}). We initially consider the case of a single ideal measurement and then develop the model to consider imperfect and repeated measurements. The final model is described in Subsection~\ref{sec:real}.
\subsection{Ideal Case}
Let $\cH$ be the system's Hilbert space with $\rho_1$ a density matrix denoting the initial state of the system at step $k=1$. We consider the evolution $\rho_1\mapsto \rho_2$ of such a quantum system under discrete-time quantum jumps (see e.g.~\cite[Ch. 4]{Haroche2006} or~\cite[Ch. 2]{Gardiner2000}).

Consider a set of Kraus operators $M_q:\cH\to\cH$, $q \in \{1,2\ldots,\mi\}$ that satisfy $\sum_{q=1}^{\mi} M_{q}^\dagger M_{q} = \bid$. Here we assume that there are $\mi\in\bbN$ possible quantum jumps  and $\bid$ is the identity operator on $\cH$. The  superscript {\tt id} stands for an abbreviation of ideal. Consider  an ideal world  with full access to all quantum jumps via a complete and ideal  set of  jump detectors. When the quantum jump indexed by $q$ is detected,  the state of the system changes to
\begin{equation}\label{eqn:rec1}
\rho_2 = \cM_q(\rho_1) \triangleq \frac{M_q\rho_1 M_q^\dagger}{\bbP[q]}
.
\end{equation}
Moreover,
\begin{equation}\label{eqn:recProb}
\bbP[q] = \tr{M_q \rho_1 M_q^\dagger}
\end{equation}
is the probability to   detect jump $q$, knowing the state $\rho_1$. We now consider the case of realistic experiments with possible measurement errors.

\subsection{Realistic Case (with imprecise measurements)}\label{sec:realSchPic}
We consider that the ideal detection of the jump $q$ corresponds to an ideal measure outcome $\mui=q$.
We denote by random variable $\mui\in \{1,2,\ldots,\mi\}$ this outcome provided by ideal sensors.
We assume that  realistic  sensors  provide an outcome  $\mur$ that is a random variable in the set
$ \{1,2,\ldots,\mr\}$.

 We assume that, with a known probability, an ideal measurement outcome $\mui$  occurs effectively whereas  the realistic sensors  detect  an outcome $\mur$. The correlations between the events $\mui=q$ and $\mur=p$ are modeled by classical probabilities through a stochastic matrix $\eta\in \bbR^{\mr\times \mi}$:
$$
\eta_{p,q}= \bbP[\mur=p|\mui=q].
$$
It gives the probability that the  real sensors detect  $\mur=p$ given the  ideal  sensors  would detect  $\mui=q$, $\textrm{for }p\in \{1,\ldots,\mr\},~q\in \{1,\ldots,\mi\}$. Since $\eta_{p,q}\geq 0$ and for each $q$, $\sum_{p=1}^{\mr} \eta_{p,q}=1$,
the  matrix $\eta=(\eta_{p,q}) $ is  a left  stochastic matrix.

\subsection{Realistic Experiment with Repeated Measurements}\label{sec:real} Consider the case of a sequence of discrete-time measurements.  We denote by $\rho_{k}$ the state of the system at discrete time-step $k$. Also suppose $M_{q;k}$ is the Kraus operator corresponding to the $k^{th}$ ideal measurement for  $q\in \{1,2,\ldots,\mi\}$. Note that we allow for a different set of Kraus operators $M_{\cdot;k}$ for different time-steps $k$. One can also consider $\mi$ and $\mr$ be dependent on $k$.

Similar to the previous subsection, we denote by $\mur_k\in \{1,\ldots,\mr\}$ and  $\mui_k\in \{1,\ldots,\mi\}$, the random variables corresponding to the $k^{th}$ realistic and ideal  outcomes, respectively. Therefore,
$$
\bbE[\rho_{k+1}|\rho_k,\mui_k=q] = \cM_{q;k}(\rho_k) \triangleq \frac{M_{q;k} \rho_k M_{q;k}^\dag}{\tr{M_{q;k} \rho_k M_{q;k}^\dag}}
$$
and $ \bbP\big[\mui_k=q~\big|~\rho_k \big] = \tr{M_{q;k} \rho_k M_{q;k}^\dag}$.

Also, we assume $\eta^k\in \bbR^{\mr\times \mi}$, the stochastic matrix determining the probability of error, can  depend on the discrete-time step $k$. In particular, we have
$$
\bbP[\mur_k=p|\mui_k=q]= \eta^k_{p,q},
$$
$\textrm{for }p\in \{1,\ldots,\mr\},~q\in \{1,\ldots,\mi\}$.

\section{Recursive Equation for the optimal filter}\label{sec:recEqn}
We wish to obtain a recursive equation for the optimal estimate $\hrho_{k+1}$ of the state $\rho_{k+1}$ knowing initial value $\rho_1$ and  real measurement outcomes $\mur_1,\ldots,\mur_{k}$.  This  optimal estimate $\hrho_k$ is defined as
$$
\hrho_{k} = \bbE[\rho_{k}|\rho_1,\mur_1,\ldots,\mur_{k-1}].
$$

The following theorem says that we can  ignore the original state $\rho_k$ and only consider  $\hat\rho_k$
that is shown  to be the state of a Markov process.
\begin{theorem}\label{the:mainRes}The optimal estimate $\hrho_k$ satisfies the following recursive equation
\begin{equation}\label{eqn:mainRes}
\hrho_{k+1} = \frac{\sum_{q=1}^{\mi} \eta^k_{p_k,q} M_{q;k} \hrho_{k}M_{q;k}^\dag }{ \tr{\sum_{q=1}^{\mi} \eta^k_{p_k,q} M_{q;k} \hrho_{k}M_{q;k}^\dag }},
\end{equation}
if $\mur_k = p_k$.
Moreover, we have
\begin{align}
&\bbP\Big[\mur_k=p_k\Big|\rho_1,\mur_1=p_1,\ldots, \mur_{k-1}=p_{k-1}\Big]\nonumber\\
    &\qquad =  \tr{\sum_{q=1}^{\mi} \eta^k_{p_k,q} M_{q;k} \hrho_{k}M_{q;k}^\dag}.\label{eqn:mainResProba}
\end{align}
\end{theorem}
\begin{remark}
The division in~\eqref{eqn:mainRes} by the R.H.S of~\eqref{eqn:mainResProba} could appear problematic  when this denominator vanishes. Nevertheless, if we assume that the real measurements are $\mur_1=p_1,\ldots,\mur_k=p_k$, then $$\bbP\Big[ \mur_1=p_1,\ldots, \mur_{k-1}=p_{k-1} \Big|\rho_1\Big] >0$$ and $$\bbP\Big[ \mur_1=p_1,\ldots, \mur_k=p_k \Big|\rho_1\Big] >0$$  (otherwise such measurement outcomes are not possible). Consequently,
\begin{align*}
&\bbP\Big[\mur_k=p_k\Big|\rho_1,\mur_1=p_1,\ldots, \mur_{k-1}=p_{k-1}\Big]\\
&\qquad = \tfrac{\bbP\Big[ \mur_1=p_1,\ldots, \mur_{k-1}=p_{k-1},\mur_k=p_k \Big|\rho_1\Big]}{\bbP\Big[ \mur_1=p_1,\ldots, \mur_{k-1}=p_{k-1} \Big|\rho_1\Big]}
\end{align*}
cannot vanish. Thus recurrence~\eqref{eqn:mainRes} is always well defined because we have~\eqref{eqn:mainResProba}.
\end{remark}
\begin{remark}[Markov property of the filter] Equations~\eqref{eqn:mainRes} and~\eqref{eqn:mainResProba} tell us that the joint-process $(\mur_k,\hrho_k)$ is a Markov process and therefore the statistics of the measurement process $\mur_k$ may be determined using $\hrho_k$. This in particular implies that we may use Monte Carlo methods to simulate the observation process $\mur_k$ only using $\hrho_k$ independent of the actual state $\rho_k$ and measurement history $\mur_1,\ldots,\mur_{k-1}$.
\end{remark}
\begin{IEEEproof}
In this proof we use the following notation for ease of presentation: we use $\mui_\i=q_\i$ and $\mur_\i=p_\i$ to denote the set of events $\left\{\mui_1=q_1,\mui_2 = q_2,\ldots,\mui_k =q_k\right\}$ and $\left\{\mur_1=p_1,\mur_2 = p_2,\ldots,\mur_k = p_k\right\}$, respectively.   For instance, using this notation, we have
{\small\begin{eqnarray*}
\lefteqn{\bbP\Big[\mui_\i = q_\i\Big|\rho_1,\mur_\i = p_\i\Big]}\\ &\triangleq& \bbP\Big[\mui_1=q_1,\ldots,\mui_k=q_k\Big|\rho_1,\mur_1=p_1,\ldots, \mur_k=p_k\Big].
\end{eqnarray*}
}
Assume that the  values measured by the real detector are $p_1=\mur_1$, \ldots, $p_k=\mur_{k}$.
Then we have the optimal estimate
{\small
\begin{multline}\label{eq:rhohat}
\lefteqn{\hrho_{k+1} =}\\ \sum_{q_1,\ldots,q_k=1}^{\mi}
\bbP\Big[\mui_\i = q_\i\Big|\rho_1,\mur_\i = p_\i \Big]\cM_{q_k;k}(\ldots(\cM_{q_1;1}(\rho_1)) \ldots ),
\end{multline}}
where
\begin{align}
&{\cM_{q_k;k} (\cdots (\cM_{q_1;1}(\rho_1)) \cdots )}\nonumber\\
&\qquad=\frac{M_{q_k;k} \cdots M_{q_1;1} \rho_1 M_{q_1;1}^\dag \cdots M_{q_k;k}^\dag}
  {\tr{M_{q_k;k} \cdots M_{q_1;1} \rho_1 M_{q_1;1}^\dag \cdots M_{q_k;k}^\dag}}.\label{eq:mqq}
\end{align}
Using Bayes law, we have for each $(q_1,\ldots,q_k)$,
\begin{multline}\label{eq:BigBayes}
    \bbP\Big[\mui_\i = q_\i\Big|\rho_1,\mur_\i = p_\i\Big]
      \bbP\Big[ \mur_\i=p_\i\Big| \rho_1\Big]
    \\=
     \bbP\Big[\mur_\i=p_\i\Big|\rho_1,\mui_\i=q_\i\Big]
      \bbP\Big[ \mui_\i=q_\i\Big| \rho_1\Big],
\end{multline}
where
\begin{align*}
&\bbP\Big[ \mui_\i=q_\i\Big| \rho_1\Big]   = \tr{M_{q_k;k} \ldots M_{q_1;1} \rho_1 M_{q_1;1}^\dag \ldots M_{q_k;k}^\dag}
\end{align*}
and
\begin{align*}
\bbP\Big[\mur_\i=p_\i\Big|\rho_1,\mui_\i=q_\i\Big] &=
\bbP\Big[\mur_\i=p_\i\Big|\mui_\i=q_\i\Big] \\
&= \eta^1_{p_1,q_1} \cdots \eta^k_{p_k,q_k}.
\end{align*}
Summing~\eqref{eq:BigBayes} over all $(q_1,\ldots,q_k)$ gives:
\begin{multline} \label{eq:Preal}
\bbP\Big[ \mur_\i =p_\i\Big| \rho_1 \Big] =
\sum_{s_1,\ldots,s_k=1}^{\mi} \eta^1_{p_1,s_1} \ldots \eta^k_{p_k,s_k}\times\\
\tr{M_{s_k;k} \ldots M_{s_1;1} \rho_1 M_{s_1;1}^\dag \ldots M_{s_k;k}^\dag}.
\end{multline}

Consequently, we have

{\scriptsize{
\begin{align*}
&\bbP\Big[\mui_\i=q_\i\Big| \rho_1,\mur_\i=p_\i\Big]
\\
&=\frac{\eta^1_{p_1,q_1} \ldots \eta^k_{p_k,q_k}\tr{M_{q_k;k} \ldots M_{q_1;1} \rho_1 M_{q_1;1}^\dag \ldots M_{q_k;k}^\dag}}
{\tr{ \sum_{s_1,\ldots,s_k=1}^{\mi} \eta^1_{p_1,s_1} \ldots \eta^k_{p_k,s_k}
M_{s_k;k} \ldots M_{s_1;1} \rho_1 M_{s_1;1}^\dag \ldots M_{s_k;k}^\dag
}}.
\end{align*}}}
Injecting the above relation into~\eqref{eq:rhohat} and using~\eqref{eq:mqq} yields
{\scriptsize
\begin{multline}
    \hrho_{k+1} =\\
    \frac{\sum_{q_1,\ldots,q_k} \eta^1_{p_1,q_1} \ldots \eta^k_{p_k,q_k}
M_{q_k;k} \cdots M_{q_1;1} \rho_1 M_{q_1;1}^\dag \cdots M_{q_k;k}^\dag}
{\tr{ \sum_{q_1,\ldots,q_k} \eta^1_{p_1,q_1} \ldots \eta^k_{p_k,q_k}
M_{q_k;k} \cdots M_{q_1;1} \rho_1 M_{q_1;1}^\dag \ldots M_{q_k;k}^\dag}}.
\label{eq:rhohatbis}
\end{multline}  }
It is then clear that $\hrho_{k+1}$ can be calculated from $\rho_1$ in a recursive manner according to
\begin{eqnarray*}
&\hrho_{k+1}= \frac{\sum_{q}  \eta^k_{p_k,q} M_{q;k}\hrho_{k}  M_{q;k}^\dag }
          {\tr{ \sum_{q}  \eta^k_{p_k,q} M_{q;k}\hrho_{k}  M_{q;k}^\dag}},&\\ &\vdots&\\
&\hrho_{2}= \frac{\sum_{q}  \eta^1_{p_1,q} M_{q;1}\hrho_{1}  M_{q;1}^\dag }
          {\tr{ \sum_{q}  \eta^1_{p_1,q} M_{q;1}\hrho_{1}  M_{q;1}^\dag}}&,
\end{eqnarray*}
where the intermediate states correspond to optimal estimates  from step $2$ to $k$:
 $\hrho_j =  \bbE[\rho_j|\rho_1,\mur_1=p_1,\ldots,\mur_{j-1}=p_{j-1}]$, $j=2,\ldots, k$. The recursive relation~\eqref{eqn:mainRes} is thus proved.

We now prove~\eqref{eqn:mainResProba}. In the following, we set ordered product
\begin{eqnarray*}
\vec{\bM}_{j} &\triangleq& M_{q_j;j}\cdots M_{q_2;2}\cdot M_{q_1;1}
\end{eqnarray*}
for $j=1,\ldots,k$.

We have
\begin{align*}
&\bbP\Big[\mur_k=p_k\Big|\rho_1,\mur_1=p_1,\ldots, \mur_{k-1}=p_{k-1}\Big]\\
&\qquad=\tfrac{\bbP\Big[\mur_1=p_1,\ldots, \mur_{k}=p_{k}\Big|\rho_1\Big]}
{\bbP\Big[\mur_1=p_1,\ldots, \mur_{k-1}=p_{k-1}\Big|\rho_1\Big] }.
\end{align*}
This fraction can be computed  using~\eqref{eq:Preal}:
{\scriptsize
\begin{align*}
&\bbP\Big[\mur_k=p_k\Big|\rho_1,\mur_1=p_1,\ldots, \mur_{k-1}=p_{k-1}\Big]\\
&\qquad= \sum_{q_1,\ldots,q_k} \eta^1_{p_1,q_1} \ldots \eta^k_{p_k,q_k}
\tr{\vec{\bM}_{k} \rho_1 \vec{\bM}_{k}^\dag}\times\\
&\qquad\left(\sum_{q_1,\ldots,q_{k-1}} \eta^1_{p_1,q_1} \ldots \eta^{k-1}_{p_{k-1},q_{k-1}}
\tr{\vec{\bM}_{k-1}\rho_1\vec{\bM}_{k-1}^\dag}
        \right)^{-1}.
\end{align*}}
According to~\eqref{eq:rhohatbis} with $k-1$ instead of $k$, we have
{\scriptsize
\begin{align*}
&\hrho_{k} = \sum_{q_1,\ldots,q_{k-1}} \eta^1_{p_1,q_1} \ldots \eta^k_{p_{k-1},q_{k-1}}
\vec{\bM}_{k-1} \rho_1 \vec{\bM}_{k-1}^\dag\times\\
&\qquad\tr{\sum_{q_1,\ldots,q_{k-1}} \eta^1_{p_1,q_1} \ldots \eta^{k-1}_{p_{k-1},q_{k-1}}
\tr{\vec{\bM}_{k-1}\rho_1\vec{\bM}_{k-1}^\dag}}^{-1}.
\end{align*}
}
Since
{\scriptsize
\begin{align*}
&\sum_{q_1,\ldots,q_k} \eta^1_{p_1,q_1} \ldots \eta^k_{p_k,q_k}
\tr{\vec{\bM}_{k} \rho_1 \vec{\bM}_{k}^\dag} =  \sum_{q_k} \eta^k_{p_k,q_k}\times\\
&\tr{ M_{q_k;k} \left(\sum_{q_1,\ldots,q_{k-1}} \!\!\!\!\!\!\eta^1_{p_1,q_1} \ldots \eta^{k-1}_{p_{k-1},q_{k-1}}
\vec{\bM}_{k-1} \rho_1 \vec{\bM}_{k-1}^\dag\right)  M_{q_k;k}^\dag },
\end{align*}}
we get  finally~\eqref{eqn:mainResProba}.

 \end{IEEEproof}

\section{Stability with respect to initial conditions}\label{sec:stab}

Assume that we do not have access to the real initial state $\rho_1$. We cannot compute the  optimal estimate $\hrho_k$. We can still use the recurrence formula~\eqref{eqn:mainRes} based on the real measurement outcomes  $(\mur_{j}=p_{j})_{j=1,\ldots,k-1}$  to propose an estimation $\rhoe_k$ of $\rho_k$. We will prove  below that this estimation procedure is stable in the sense that the fidelity between $\hrho_k$ and $\rhoe_k$ is non-decreasing in average whatever the initial condition $\rhoe_1$ is.

For ease of notation we set
\begin{align}
M_{p,q;k} &= \sqrt{\eta^k_{p,q}}M_{q;k}, \notag
\\
{\cM}_{p;k}(\rho) &= \frac{\sum_{q=1}^{\mi} M_{p,q;k} \rho M_{p,q;k}^\dag }{\sum_{q=1}^{\mi} \tr{M_{p,q;k} \rho M_{p,q;k}^\dag }}
\label{eq:cMp}
\end{align}
for any $k\geq 1$, $p\in\{1,\ldots,\mr\}$ and $q\in\{1,\ldots,\mi\}$.
The sets $$\cS_{p;k} \triangleq \{M_{p,1;k},\ldots,M_{p,\mi;k}\}$$ for $p=1,\ldots,\mr$ form a partition of
$$\cS_k \triangleq \{M_{p,q;k}~|~ p=1,\ldots,\mr,~q=1,\ldots,\mi\}.
$$
Using this notation, the recursive Equation~\eqref{eqn:mainRes} defines a coarse-grained Markov chain in the sense of~\cite{Rouchon2011}.

If $\mur_k = p$,  we define $\rhoe_k$ recursively for $k\geq 1$ as follows
\begin{equation}\label{eqn:recRel}
\rhoe_{k+1}= \cM_{p;k}\big(\rhoe_{k}\big)
.
\end{equation}
Such a recursive formula is valid as soon as $\sum_{q} \tr{M_{p,q;k} \rhoe_{k}M_{p,q;k}^\dag }>0$, which is automatically satisfied when $\rhoe_k$ is of full rank. $\rhoe_{k+1}$ is indeterminate when $\sum_{q} \tr{M_{p,q;k} \rhoe_{k}M_{p,q;k}^\dag }=0$. But using the continuity arguments  developed at the end of the appendix we can give a value for $\rhoe_{k+1}$ in the following way: for each $\rhoe_k$, consider the set of density operators
$$
\left( \cM_{p;k}\big(\rhoe_{k,\epsilon}\big) \right)_{p=1,\ldots,\mr},
$$
where $\epsilon >0$ and $\rhoe_{k,\epsilon}=(\rhoe_k + \epsilon \bid)/\tr{\rhoe_k + \epsilon \bid}$. Since $\rhoe_{k,\epsilon}$ is positive definite, $\left( \cM_{p;k}\big(\rhoe_{k,\epsilon}\big) \right)_{p=1,\ldots,\mr}$ are well defined and admit a limit point when $\epsilon$ tends to $0^+$ ($\cH$ is of finite dimension here). Take for each $\rhoe_k$ such a limit point $(\rhoe_{k+1,p})_{p=1,\ldots,\mr}$. Set $\rhoe_{k+1}= \rhoe_{k+1,p}$ when $\mur_k = p$. If $\sum_{q} \tr{M_{p,q;k} \rhoe_{k}M_{p,q;k}^\dag }>0$ then $\rhoe_{k+1,p}$ coincides necessarily with $\cM_{p;k}\big(\rhoe_{k}\big)$ and we recover~\eqref{eqn:recRel}.  During the proof of Theorem~\ref{the:stab}, we will use only recurrence~\eqref{eqn:recRel} having in mind that, when $\sum_{q} \tr{M_{p,q;k} \rhoe_{k}M_{p,q;k}^\dag }=0$ we have to use $\rhoe_{k+1}=\rhoe_{k+1,p}$.

If  the initial state of the system $\rhoe_1$ coincides with $\rho_1$, then $\rhoe_k$ coincides with  the optimal estimate $\hrho_k$ of the state $\rho_k$ from Theorem~\ref{the:mainRes}. In fact, once the initial states $\hrho_1$ and $\rhoe_1$ are given,  the process $(\hrho_k)$ and $(\rhoe_k)$ are  driven by the same  stochastic process $(\mur_{k-1})$, itself   driven by the combination of  the original quantum process of state $(\rho_{k})$ with  the classical  process associated to the  left stochastic matrices $(\eta^k)$ governing detection errors.

 The following theorem shows that the fidelity between $\hrho_k$ and  $\rhoe_k$ is non-decreasing in average.
\begin{theorem}\label{the:stab} Suppose $\rhoe_k$ satisfies the recursive relation~\eqref{eqn:recRel} with an arbitrary initial density operator $\rhoe_1$.  Then the fidelity between $\hrho_k$ and $\rhoe_k$ defined by
$$
F\big(\hrho_k,\rhoe_k\big)\triangleq \left(\tr{\sqrt{\sqrt{\hrho_k}\rhoe_k\sqrt{\hrho_k}}}\right)^2
.
$$
is a submartingale in the following sense:
$$\bbE\bigg[F\big(\hrho_{k+1},\rhoe_{k+1}\big)\bigg|\hrho_1,\ldots,\hrho_k,\rhoe_1,\ldots,\rhoe_k\bigg] \geq F\big(\hrho_k,\rhoe_k\big).$$
\end{theorem}
\begin{proof}
Denote $F_k=F\big(\hrho_k,\rhoe_k\big)$.  We have
\begin{eqnarray*}
\lefteqn{\bbE\bigg[F_{k+1}\bigg|\hrho_1,\ldots,\hrho_k,\rhoe_1,\ldots,\rhoe_k\bigg]}\\
 &=&\bbE\bigg[F_{k+1}\bigg|\hrho_1,\rhoe_1, \mur_1=p_1,\ldots, \mur_{k-1}=p_{k-1}\bigg] \\
&=& \sum_{p=1}^{\mr} \bbP\bigg[ \mur_k=p \bigg| \hrho_1,\rhoe_1,\mur_1=p_1,\ldots, \mur_{k-1}=p_{k-1}\bigg]\times \\
  && \bbE\bigg[F_{k+1}\bigg|  \hrho_1,\rhoe_1, \mur_1=p_1,\ldots, \mur_{k-1}=p_{k-1},\mur_{k}=p\bigg].
\end{eqnarray*}
The conditional probabilities appearing in this sum are given by~\eqref{eqn:mainResProba} and the conditional expectations read
\begin{multline*}
{\bbE\bigg[F_{k+1}\bigg|\hrho_1,\rhoe_1,\mur_1=p_1,\ldots, \mur_{k-1}=p_{k-1},\mur_{k}=p\bigg]}\\
\qquad=F\left({\cM}_{p;k}(\hrho_k)~,~{\cM}_{p;k}(\rhoe_k)\right)
\end{multline*}
since, once $\hrho_1$ and $\rhoe_1$ and $\mur_1=p_1$,\ldots, $\mur_{k-1}=p_{k-1}$ and $\mur_{k}=p$ are given, $\hrho_{k+1}={\cM}_{p;k}(\hrho_k)$ and
$\rhoe_{k+1}={\cM}_{p;k}(\rhoe_k)$. Thus we have
\begin{align*}
&\bbE\bigg[F_{k+1}\bigg|\hrho_1,\ldots,\hrho_k,\rhoe_1,\ldots,\rhoe_k\bigg]=
 \\
&\qquad\sum_{p=1}^{\mr} \left(\text{\scriptsize $\sum_q \tr{M_{p,q;k} \hrho_k M_{p,q;k}^\dag}$}  \right)\times\\
& \qquad F\left(\tfrac{\sum_{q} M_{p,q;k} \hrho_k M_{p,q;k}^\dag }{\sum_{q} \tr{M_{p,q;k} \hrho_k M_{p,q;k}^\dag }}~,~
  \tfrac{\sum_{q} M_{p,q;k} \rhoe_k M_{p,q;k}^\dag }{\sum_{q} \tr{M_{p,q;k} \rhoe_k M_{p,q;k}^\dag }}\right)
  .
\end{align*}
The fact that $\bbE\bigg[F_{k+1}\bigg|\hrho_1,\ldots,\hrho_k,\rhoe_1,\ldots,\rhoe_k\bigg] \geq F_k$ is then a direct consequence of equation~\eqref{eq:appendix} given in appendix with  $r=\mr$, $s=\mr \mi$,  index $j$ corresponding to $p$, index $i$ to $(p,q)$, operators $L_i$ to $M_{p,q;k}$, density operators $\rho$ and $\sigma$ to $\hrho_k$ and $\rhoe_k$, respectively.

\end{proof}

\section{Example: Quantum Filter for the Photon-Box}\label{sec:example}
This section considers, as a key  illustration,  the  quantum filter design  in~\cite{Sayrin2011} to estimate in real-time  the state  $\rho$  of a quantized electro-magnetic field. Since this filter admits exactly the recursive form of Theorem~\ref{the:mainRes}, Theorem~\ref{the:stab} applies and thus, this filter is proved here to be stable versus its initial condition. 

The actual experiment under consideration uses quantum non-demolition measurements~\cite{Haroche2006} to estimate the state of the quantized field trapped in a superconducting microwave cavity. Circular Rydberg atoms are sent at discrete time intervals to perform partial measurements of the photon number. Atoms are subsequently detected either in their excited (e) or ground (g) state. The outcomes of these measurements are then used to estimate the state of the cavity field, thanks to the quantum filter described below. This estimation is eventually used to calculate the amplitude of classical fields injected in the cavity in order to bring the field closer to a predefined target state. The interested reader is directed to~\cite{Dotsenko2009} and~\cite{Sayrin2011} for further details of the experimental setup and results obtained.

The Hilbert space $\cH$ of the cavity is,  up to some  finite photon number truncation (around $10$), the Fock space
with basis $\{\ket{n}\}_{n\geq 0}$, each $\ket{n}$ being  the Fock state associated to exactly $n$ photons (photon-number state). The annihilation operator $\ba:\cH\mapsto \cH$ is defined by
$\ba \ket{n} = \sqrt{n} \ket{n-1}$ for $n\geq 1$ and $\ba\ket{0}=0$. Its Hermitian conjugate $\ba^\dag$ is the creation operator satisfying $\ba \ket{n} =\sqrt{n+1} \ket{n+1}$, for all $n\geq 0$. The photon-number operator (energy operator) is $\bn=\ba^\dag\ba$ which satisfies $\bn \ket{n} = n \ket{n}$, for all $n\geq 0$. Recall the commutation $[\ba,\ba^\dag]= \bid$.

In~\cite{Sayrin2011} the following imperfections have been considered:
\begin{itemize}

\item atomic preparation efficiency  characterized by $P_a(n_a) \geq 0 $,  the probability to have $n_a\in\{0,1,2\}$ atom(s) interact with the cavity: $P_a(0)+P_a(1)+P_a(2)=1$.

\item detection efficiency characterized by $\epsilon_d\in[0,1]$, the probability that the detector detects an atom when it is present.

\item state detection error rate $\eta_g\in[0,1]$ (resp. $\eta_e$)   probability of erroneous state
assignation to $e$ (resp. $g$)  when the atom collapses in  $g$ (resp. $e$).

\end{itemize}

The original  state $\rho$ is subject to  $\mi=3\times 7 $ possible quantum jumps and the available  sensors (atomic detector) admits only  $\mr=6$  possibilities. We begin by introducing some operators that are used to describe these quantum jumps:
\begin{align*}
   &D_\alpha = e^{\alpha \ba^\dag - \alpha^* \ba},\quad
   L_{no} =  \sqrt{P_a(0)}~\bid,\quad
   L_{g} = \sqrt{P_a(1)} \cos \phi_\bn,
   \\&
    L_{e} =\sqrt{P_a(1)} \sin \phi_\bn,\quad
    L_{gg} = \sqrt{P_a(2)} \cos^2 \phi_\bn,
    \\&
    L_{ge}=L_{eg}= \sqrt{P_a(2)} \cos \phi_\bn \sin \phi_\bn,\quad
    L_{ee} = \sqrt{P_a(2)} \sin^2 \phi_\bn,
    \\&
     L_{o}= 1 - \tfrac{\epsilon(1+2 n_{th})}{2}  \bn -  \tfrac{\epsilon n_{th}}{2}\bid,\quad
     L_{+} = \sqrt{\epsilon(1+n_{th})} \ba,
   \\&
    L_{-} = \sqrt{\epsilon n_{th}} \ba^\dag,
\end{align*}
where $\phi_\bn =\frac{\phi_0(\bn+1/2)+\phi_R}{2}$ and  $0<\epsilon, n_{th} \ll 1$, $\phi_0,\phi_R$ are real experimental parameters.
The unitary displacement operator $D_\alpha$  corresponds to the  control input  $\alpha\in\mathbb C$, depending on the sampling step $k$.
The    operators $L_{o}$, $L_+$ and $L_-$ correspond to the interaction of the cavity-field  with its environment (decoherence due to mirrors and   thermal photons):
\begin{enumerate}
\item $L_{o}$  corresponds to no photon jump;
\item$L_+$ corresponds to the capture  of  one thermal photon by the cavity-field;
\item $L_-$ corresponds to one photon lost  from  the cavity-field.
\end{enumerate}
Since $L_{o}^\dag L_{o} + L_{+}^\dag L_{+} + L_{-}^\dag L_{-}=\bid + O(\epsilon^2)$ and $\epsilon$ is small, we consider  in the sequel that $(L_{o},L_+,L_-)$ are associated to an effective Kraus map $ L_{o} \rho L_{o}^\dag  + L_{+} \rho L_{+}^\dag + L_{-} \rho L_{-}^\dag$.

The operators $L_{no}$, $L_{g}$, $L_e$, $L_{gg}$, $L_{ge}$, $L_{eg}$ and $L_{ee}$  correspond to the jump induced by the collapse  of  possible   crossing atom(s)  having interacted with the cavity-field:
\begin{enumerate}
\item $L_{no}$ - no  atom in the atomic sample;
\item $L_g$ - one atom having interacted with the cavity-field and collapsed to the atomic ground state during the detection process;
\item $L_e$ - one atom having interacted with the cavity-field and collapsed to the atomic excited  state during the detection process;
\item $L_{gg}$ - two atoms having interacted with the cavity-field, both having collapsed  to $g$;
\item $L_{ge}$ - two atoms having interacted with the cavity-field,   the first one  having collapsed to $g$ and the second   to $e$;
\item $L_{eg}$ - two atoms having interacted with the cavity-field,   the first one  having collapsed to $e$ and the second to $g$;
\item $L_{ee}$ - two atoms having interacted with the  cavity-field,  both having  collapsed  to $e$.
\end{enumerate}

For each control input $\alpha$, we have a total of $\mi=3\times 7$  Kraus operators. The jumps are labeled by  $q=(q^a,q^c)$ with $q^a\in \{no,g,e,gg,ge,eg,ee\}$ labeling atom related jumps  and $q^c\in\{o,+,-\}$ cavity decoherence jumps.  The Kraus operators associated to such $q$  are $M_q = L_{q^c} D_\alpha L_{q^a}$. So, for instance, with the  control input $\alpha_k$ at step $k$:
\begin{itemize}
\item the Kraus operator corresponding to one  atom collapsing in ground state, $q^a=g$, and   one photon  destroyed by mirrors, $q^c=-$, reads  $M^k_{(g,-)}= L_{-}D_{\alpha_k} L_{g}$.
\item the Kraus operator corresponding to two   atoms, the first one  collapsing to $g$, the second one to $e$, $q^a=ge$, and   one thermal photon  being caught between the two mirrors, $q^c=+$, reads  $M^k_{(ge,+)}= L_{+}D_{\alpha_k} L_{ge}$.
\end{itemize}
One can  check that, for any  value of $\alpha$,   these 21  operators define a Kraus map (using the assumption that $L_{o}^\dag L_{o} + L_{+}^\dag L_{+} + L_{-}^\dag L_{-}\approx \bid$).

We have only $\mr=6$ real detection possibilities $p\in\{no,g,e,gg,ge,ee\}$ corresponding respectively to no detection, a single detection in $g$, a single detection  in $e$, a double detection both  in $g$, a double detection one in $g$ and the other in $e$, and a double  detection both  in $e$. A double detection does not distinguish  two atoms.   The entries $\eta_{p,q}$ of the stochastic matrix describing the imperfect detections  corrupted by errors  are independent of $\alpha$ and  given by Table \ref{tb:stochastic}.   It relies on  error rate parameters $\eta_g$, $\eta_e$ and $\epsilon_d$ in $[0,1]$ and assume no  error correlation  between  atoms. So for instance the probability that there is a single atom in the sample, which collapses to the ground state and is in fact detected by the experimental sensors to be in the ground state is given by $\eta_{g,(g,q^c)} = \epsilon_d(1-\eta_g)$. Note that each column in the table sums up to one, since $\eta_{p,q}$ is a left stochastic matrix. 
    \begin{table*}[ht]
    \centerline{
    \begin{tabular}{|c||c|c|c|c|c|c|}
      \hline
    $p\setminus q$
          &$(no,q^c)$ & $(g,q^c)$ & $(e,q^c)$ &$(gg,q^c)$ &  $(ee,q^c)$ & $(ge,q^c)$ or $(eg,q^c)$
    \\\hline\hline
     $no$ & $1$& $1-\epsilon_d$& $1-\epsilon_d$  &  $(1-\epsilon_d)^2$& $(1-\epsilon_d)^2$&  $(1-\epsilon_d)^2$
    \\\hline
     $g$  & $0$& $\epsilon_d(1-\eta_g)$& $\epsilon_d\eta_e$ &  $2\epsilon_d(1-\epsilon_d)(1-\eta_g)$&  $2\epsilon_d(1-\epsilon_d)\eta_e$&   $\epsilon_d(1-\epsilon_d)(1-\eta_g+\eta_e)$
    \\\hline
     $e$  & $0$& $\epsilon_d\eta_g$&  $\epsilon_d(1-\eta_e)$&  $2\epsilon_d(1-\epsilon_d)\eta_g$& $2\epsilon_d(1-\epsilon_d)(1-\eta_e)$&   $\epsilon_d(1-\epsilon_d)(1-\eta_e+\eta_g)$
    \\\hline
     $gg$ & $0$& $0$& $0$ & $\epsilon_d^2(1-\eta_g)^2$& $\epsilon_d^2\eta_e^2$&  $\epsilon_d^2\eta_e(1-\eta_g)$
    \\\hline
     $ge$ & $0$& $0$& $0$&  $2\epsilon_d^2\eta_g(1-\eta_g)$& $2\epsilon_d^2\eta_e(1-\eta_e)$&   $\epsilon_d^2((1-\eta_g)(1-\eta_e)+\eta_g\eta_e)$
    \\\hline
     $ee$ & $0$& $0$& $0$&   $\epsilon_d^2\eta_g^2$& $\epsilon_d^2(1-\eta_e)^2$& $\epsilon_d^2\eta_g(1-\eta_e)$
    \\\hline
    \end{tabular}}
    \vline
    \caption{Stochastic matrix $\eta_{p,q}$ .}
    \label{tb:stochastic}
    \end{table*}

\section{Conclusion}
In this paper, we derive a recursive expression for the optimal estimate of a quantum system's state from imperfect, discrete-time measurements. The optimality of this  recursive expression  is proven by a simple application of Bayes' lemma and quantum measurement postulates. Such a filter is shown to satisfy a Markov property and thus can  be used  for quantum control as shown in~\cite{Sayrin2011} or to run Monte Carlo simulations of the measurement trajectories. We also demonstrate that this filter is stable with respect to initial conditions in the sense of Theorem~\ref{the:stab}. In the future, the continuous-time version of these results will be investigated.
\bibliographystyle{plain}
\bibliography{ref}

\appendix
\section{An inequality}
Consider a set of $s$ operators $(L_i)_{i=1,\ldots,s}$ on the finite dimensional  Hilbert space  $\cH$,  such that $L_i\neq 0$ for all $i$ and
$\sum_{i=1}^s L_i^\dag L_i=\bid$. Take a partition of $\{1,\ldots,s\}$ into $r\geq 1$ non-empty sub-sets $\left({\mathcal P}_j\right)_{j=1,\ldots, r}$.
Then, for any semi-definite positive operators $\rho$ and $\sigma$ on $\cH$ with  unit traces, the following inequality proved in~\cite{Rouchon2011} holds true:
\begin{align}
&F( \rho,\sigma)\leq \sum_{j=1}^{r}
\text{\scriptsize $\tr{\Sigma_{i\in{\mathcal P}_{j}}L_i\rho L_i^\dag}$}\times\nonumber\\&\qquad
  F \left(
   \tfrac{\sum_{i\in{\mathcal P}_{j}}L_i\rho L_i^\dag }{\tr{\sum_{i\in{\mathcal P}_{j}}L_i\rho L_i^\dag}}~,~
     \tfrac{\sum_{i\in{\mathcal P}_{j}}L_i \sigma L_i^\dag }{\tr{\sum_{i\in{\mathcal P}_{j}}L_i \sigma L_i^\dag}}
     \right),
      \label{eq:appendix}
\end{align}
where $F(\rho,\sigma)=F(\sigma,\rho)$ is the fidelity between $\sigma$ and $\rho$ defined as
$$
F(\rho,\sigma)=\left(\tr{\sqrt{\sqrt{\sigma}\rho\sqrt{\sigma}}} \right)^2
.
$$
When,  in the above sum, a  denominator $\tr{\sum_{i\in{\mathcal P}_{j}}L_i\rho L_i^\dag} $  depending on $\rho$  vanishes, the sum remains still valid since the corresponding value of  $F$  is multiplied by the same vanishing factor and $F$ is bounded since between  $0$ and $1$. This is no more true when a denominator $\tr{\sum_{i\in{\mathcal P}_{j}}L_i\sigma  L_i^\dag} $  depending on $\sigma$ vanishes. In this case the above inequality should be interpreted in the following way relying on a continuity argument.
For each $\sigma$, define
$$
\cZ_\sigma=\left\{j ~\bigg|~\tr{\sum_{i\in{\mathcal P}_{j}}L_i\sigma  L_i^\dag}=0 \right\} \subset \{1,\ldots,r\}.
$$
For almost all density operators  $\sigma$, $\cZ_\sigma=\emptyset$. Take $\epsilon >0$ and consider the positive definite density operator  $\sigma^\epsilon=\frac{\sigma + \epsilon \bid}{\tr{\sigma + \epsilon \bid}}$.
 For $j\in\{1,\ldots,r\}$, $\tr{\sum_{i\in{\mathcal P}_{j}}L_i\sigma^\epsilon  L_i^\dag}> 0$ since each $L_i\neq 0$. In particular,
 the set  of density operators $$
 \left(\tfrac{\sum_{i\in{\mathcal P}_{j}}L_i \sigma^\epsilon L_i^\dag }{\tr{\sum_{i\in{\mathcal P}_{j}}L_i \sigma^\epsilon L_i^\dag}}\right)_{j\in\cZ_\sigma}
 $$  admits at least an accumulation point when $\epsilon$ tends to $0^+$ ($\cH$ of finite dimension implies that the set of density operators is compact). Denote by $(\sigma^+_j)_{j\in\cZ_\sigma}$, such an accumulation point where each $\sigma^+_j$ is a density operator. Since for any $\epsilon >0$, inequality~\eqref{eq:appendix} holds true for $\rho$ and $\sigma^\epsilon$, $F$ is continuous, and for any $j\notin\cZ_\sigma$, $\tfrac{\sum_{i\in{\mathcal P}_{j}}L_i \sigma^\epsilon L_i^\dag }{\tr{\sum_{i\in{\mathcal P}_{j}}L_i \sigma^\epsilon L_i^\dag}}$ tends to $\tfrac{\sum_{i\in{\mathcal P}_{j}}L_i \sigma L_i^\dag }{\tr{\sum_{i\in{\mathcal P}_{j}}L_i \sigma L_i^\dag}}$, we have by continuity
 \begin{multline*}
 F(\rho,\sigma) \geq
    \sum_{j\in\cZ_\sigma}
\text{\scriptsize $\tr{\Sigma_{i\in{\mathcal P}_{j}}L_i\rho L_i^\dag}$}
  F \left(
   \tfrac{\sum_{i\in{\mathcal P}_{j}}L_i\rho L_i^\dag }{\tr{\sum_{i\in{\mathcal P}_{j}}L_i\rho L_i^\dag}}~,~
     \sigma^+_j
     \right)
     \\
     +
         \sum_{j\notin\cZ_\sigma}
\text{\scriptsize $\tr{\Sigma_{i\in{\mathcal P}_{j}}L_i\rho L_i^\dag}$}
  F \left(
   \tfrac{\sum_{i\in{\mathcal P}_{j}}L_i\rho L_i^\dag }{\tr{\sum_{i\in{\mathcal P}_{j}}L_i\rho L_i^\dag}}~,~
     \tfrac{\sum_{i\in{\mathcal P}_{j}}L_i \sigma L_i^\dag }{\tr{\sum_{i\in{\mathcal P}_{j}}L_i \sigma L_i^\dag}}
     \right)
     .
\end{multline*}
Such continuity argument show that we can extend inequality~\eqref{eq:appendix}  via the accumulation value(s) $\sigma^+_j$ when $\cZ_\sigma$ is not empty.

\end{document}